\documentclass[twocolumn]{mysvjour3}         

\smartqed 

\usepackage{amssymb}
\usepackage{amsfonts}
\usepackage[colorlinks,citecolor=blue,urlcolor=blue,breaklinks=true]{hyperref}
\usepackage{cite}
\usepackage[cmex10]{amsmath}
\usepackage{url}
\usepackage{arydshln}

\journalname{International Journal of Information Security}

\begin{document}

\title{Perfect Secrecy Systems Immune to Spoofing Attacks
\thanks{This work was supported by the Deutsche Forschungsgemeinschaft (DFG) via a Heisenberg grant (Hu954/4) and a Heinz Maier-Leibnitz Prize grant (Hu954/5).}
}

\titlerunning{Perfect Secrecy Systems Immune to Spoofing Attacks}

\author{Michael Huber}

\authorrunning{Michael Huber}

\institute{M. Huber \at
              Wilhelm-Schickard-Institute for Computer Science\\ University of Tuebingen\\
              Sand~13, 72076~Tuebingen, Germany\\
              \email{michael.huber@uni-tuebingen.de}\\
              Phone +49 7071 2977173\\
              Fax         +49 7071 295061 
          }

\date{Received: June 24, 2011 / Revised: May 9, 2012} 

\maketitle
\begin{abstract}
We present novel perfect secrecy systems \linebreak that provide immunity to spoofing attacks under equi-\linebreak probable source probability distributions. On the theoretical side, relying on an existence result for $t$-designs by Teirlinck, our construction method constructively generates systems that can reach an arbitrary high level of security. On the practical side, we obtain, via cyclic difference families, very efficient constructions of new optimal systems that are onefold secure against spoofing. Moreover,  we construct, by means of $t$-designs for large values of $t$,  the first near-optimal systems that are $5$- and $6$-fold secure as well as further systems with a feasible number of keys that are $7$-fold secure against spoofing. We apply our results furthermore to a recently extended authentication model, where the opponent has access to a verification oracle. We obtain this way novel perfect secrecy systems with immunity to spoofing in the verification oracle model. 
\end{abstract}

\section{Introduction}\label{Introduction}

Perfect secrecy systems (or codes) play a prominent role in information theory and cryptography. In terms of information theoretic security, these systems shall ensure protection of the confidentiality of sensitive information in the presence of eavesdropping. The information theoretic, or unconditional, security model does not depend on any complexity assumptions and hence cannot be broken given unlimited computational resources. A well-known example of a perfect secrecy system is Vernam's One-time Pad. In his landmark paper ``Communication theory of secrecy systems''~\cite{Shan49}, Shannon established a fundamental characterization of optimal perfect secrecy systems: A key-minimal secrecy system achieves perfect secrecy if and only if the encryption matrix is a Latin square and the keys are used with equal probability. Important generalizations have been obtained since then (see, e.g.,~\cite{GM90,Stin90,Stin93}).
In addition to the concept of perfect secrecy, various scenarios require that the systems provide robustness against spoofing attacks. Concerning the aspect of authenticity, the integrity of information that is communicated via a potentially insecure channel shall be assured. Often such constructions involve a variety of tools from combinatorics (see, e.g.,~\cite{Hub2009,Hub2010,Hu2010now,pei06,Stin90}).

In this paper, we present novel perfect secrecy systems that provide immunity to spoofing attacks under equiprobable source probability distributions. In the past decades various perfect secrecy systems have been constructed that offer zero (like Vernam's One-time Pad) or onefold security against spoofing. Recently, in~\cite{Hub2009}, the first infinite classes of optimal perfect secrecy systems that achieve twofold security have been constructed as well as further optimal systems that offer up to $4$-fold security against spoofing under equiprobable source probability distributions. This has been \linebreak achieved by means of particular Steiner $t$-designs, e.g., the famous \mbox{$5$-$(12,6,1)$} Witt design. However, as Steiner \mbox{$t$-designs} are not known to exist for $t>5$, the level of security cannot be augmented any further via this approach.
In the present paper, we develop a more general construction method, which allows us to use \mbox{$t$-designs} for higher values of $t$ under equiprobable source probability distributions. On the theoretical side, relying on Teirlinck's existence result for $t$-designs~\cite{Teir1987}, our method constructively generates systems that can reach an arbitrary high security level. On the practical side,  by using cyclic difference families, we give very efficient constructions of new optimal systems that are onefold secure against spoofing. By employing \mbox{$t$-designs} for large values of $t$, we also present the first near-optimal systems that are $5$- and $6$-fold secure as well as further systems with a feasible number of keys that are $7$-fold secure against spoofing. Moreover, we apply our results to an extended authentication model, where the opponent has access to a verification oracle. This model, which has been recently introduced and investigated in~\cite{gold04,saf04,tonien07,tonien09}, allows a more powerful pro-active attack scenario. The opponent may send a message of the opponent's choice to the receiver and observe the receiver's response whether or not the receiver accepts it as authentic. This can be modeled in terms of a verification oracle with an online/offline variant that provides a response to a query message in the same way as the message would be accepted or not by the legitimate receiver. We obtain this way novel perfect secrecy systems with immunity to spoofing attacks in the verification oracle model.

The organization of the paper is as follows: The underlying information theoretic Shannon--Simmons model is given in Section~\ref{Model}. Section~\ref{CombStru} introduces background material on combinatorial structures that is important for our further purposes. Section~\ref{Known} presents a short overview of known constructions of perfect secrecy systems that provide robustness against spoofing attacks. In Section~\ref{NewMethod}, a general construction method is developed and we  examine the level of security from a theoretical point of view. The subsequent two sections deal then with the practical side: we give explicit constructions of optimal systems with onefold immunity to spoofing in Section~\ref{one_immune}, and of near-optimal and other feasible systems with multifold immunity in Section~\ref{multi_immune}. In Section~\ref{oracle}, we apply our constructions to the verification oracle model. The paper is concluded in Section~\ref{Conl}.


\section{The Shannon--Simmons Model}\label{Model}

We rely on the \emph{information theoretic} (or \emph{unconditional}) secrecy model developed by Shannon~\cite{Shan49}, and by Simmons (e.g.,~\cite{Sim85,Sim92}) including authentication. Our notation follows, for the most part, that of~\cite{Stin90,Mass86}.
In this model of authentication and secrecy three participants are involved: a \emph{transmitter}, a \emph{receiver}, and an \emph{opponent}.  The transmitter wants to communicate information to the receiver via a public communications channel. The receiver in return would like to be confident that any received information actually came from the transmitter and not from some opponent (\emph{integrity} of information). The transmitter and the receiver are assumed to trust each other. This is known as an \emph{authentication system} (or \emph{authentication code, $A$-code}).

In what follows, let $\mathcal{S}$ denote a set of $k$ \emph{source states} (or \emph{plaintexts}), $\mathcal{M}$ a set of $v$ \emph{messages} (or \emph{ciphertexts}), and $\mathcal{E}$ a set of $b$ \emph{encoding rules} (or \emph{keys}). Using an encoding rule $e\in \mathcal{E}$,
the transmitter encrypts a source state $s \in \mathcal{S}$ to obtain the message $m=e(s)$ to be sent over the channel. The encoding rule is an injective function from $\mathcal{S}$ to $\mathcal{M}$, and is communicated to the receiver via a secure channel prior to any messages being sent. For each encoding rule $e \in \mathcal{E}$, let $M(e):=\{e(s) : s \in \mathcal{S}\}$ denote the set of \emph{valid} messages. A received message $m$ will be accepted by the receiver as being authentic if and only if $m \in M(e)$. When this is fulfilled, the receiver decrypts the message $m$ by applying the decoding rule $e^{-1}$, where \[e^{-1}(m)=s \Leftrightarrow e(s)=m.\]
An authentication system can be represented algebrai-\linebreak cally by a $(b \times k)$\emph{-encoding matrix} with the rows indexed by the encoding rules, the columns indexed by the source states, and the entries defined by $a_{es}:=e(s)$ ($1\leq e \leq b$, $1\leq s \leq k$).

Concerning authenticity, we address the following scenario, called \emph{spoofing attack} of order $i$ (cf.~\cite{Mass86}):
Suppose that an opponent observes $i\geq 0$ distinct messages, which are sent through the public channel using the same encoding rule. The opponent then inserts a new message $m'$ (being distinct from the $i$ messages already sent), hoping to have it accepted by the receiver as authentic.
The cases $i=0$ and $i=1$ are called \emph{impersonation game} and \emph{substitution game}, respectively. These cases have been studied in detail in recent years, whereas less is known for higher orders.

For any $i$, we assume that there is some probability distribution on the set of \mbox{$i$-subsets} of source states, so that any set of $i$ source states has a non-zero probability of occurring. For simplification, we ignore the order in which the $i$ source states occur, and assume that no source state occurs more than once.
Given this probability distribution  $p_S$  on $\mathcal{S}$, the receiver and transmitter choose a probability distribution $p_E$ on $\mathcal{E}$  (called \emph{encoding strategy}) with associated independent random variables $S$ and $E$, respectively. These distributions are known to all participants and induce a third distribution, $p_M$, on $\mathcal{M}$ with associated random variable $M$. The \emph{deception probability} $P_{d_i}$ is the probability that the opponent can deceive the receiver with a spoofing attack of order $i$. Combinatorial lower bounds can be given as follows (cf.~\cite{Mass86}).

\begin{theorem}[Massey]
In an authentication system with $k$ source states and $v$ messages, for every $0 \leq i \leq t$, the deception probabilities are bounded below by
\[P_{d_i}\geq \frac{k-i}{v-i}.\]
\end{theorem}

An authentication system is called $t$\emph{-fold secure \linebreak against spoofing} if $P_{d_i}= (k-i)/(v-i)$ for all $0 \leq i \leq t$.
The following theorem (cf.~\cite{Mass86,Sch86}) establishes a combinatorial lower bound on the number of encoding rules for this kind of attack.

\begin{theorem}[Massey--Sch\"{o}bi]\label{thm_mas_sch}
If an authentication system is $(t-1)$-fold against spoofing, then the number of encoding rules is bounded below by
\[b \geq \frac{{v \choose t}}{{k \choose t}}.\]
\end{theorem}

Such a system is called \emph{optimal} if the number of encoding rules meets the lower bound with equality.

Concerning secrecy, we recall Shannon's fundamental idea of perfect secrecy (cf.~\cite{Shan49}):
An authentication system is said to have \emph{perfect secrecy} if
\[p_S(s | m)=p_S(s)\]
for every source state $s \in \mathcal{S}$ and every message $m \in \mathcal{M}$.
That is, the \emph{a posteriori} probability that the source state is $s$, given that the message $m$ is observed, is identical to
the \emph{a priori} probability that the source state is $s$.
From Bayes' Theorem follows that
\[p_S(s|m) =\frac{\sum_{\{e \in \mathcal{E}: e(s)=m\}} p_E(e)p_S(s)}{\sum_{\{e \in \mathcal{E}:m\in M(e)\}} p_E(e)p_S(e^{-1}(m))}.\]
This yields:

\begin{lemma}[Stinson]\label{frequency}
An authentication system has perfect secrecy if and only if
\[\sum_{\{e \in \mathcal{E}: e(s)=m\}} p_E(e)=\sum_{\{e \in \mathcal{E}:m\in M(e)\}} p_E(e)p_S(e^{-1}(m))\]
for every source state $s \in \mathcal{S}$ and every message $m \in \mathcal{M}$.
\end{lemma}

Therefore, if the encoding rules in a system are used with equal probability, then a given message $m$ occurs with the same frequency in each column of the encoding matrix.


\section{Combinatorial Structures}\label{CombStru}

We give in this section some background material on combinatorial structures that is important for our further purposes. Let us assume that $t \leq k \leq v$ and $\lambda$ are positive integers. 

\begin{definition}
Let $G$ be a finite additive Abelian group of order $v$. A \emph{difference family} DF$(v,k,\lambda)$ over $G$ is a family $\mathcal{F}=\{D_1,\ldots,D_l\}$ of subsets of $G$, satisfying the following properties:
\begin{enumerate}

\item[(i)] $\left| D_i \right|=k$ for all $i$ with $1\leq i \leq l$,

\smallskip

\item[(ii)]   the multiset union
\[\bigcup_{i=1}^l \{x-y : x,y \in D_i,\, x \neq y \}\]
contains every nonzero element of $G$ exactly $\lambda$ times.

\end{enumerate}

\end{definition}

The sets $D_1,\ldots,D_l$ are called \emph{base blocks}. A difference family with a
single base block is called a \emph{difference set}. A DF$(v,k,\lambda)$ with $G$ isomorphic to the cyclic group $C_v$ of order $v$ is called a \emph{cyclic} difference family and denoted by CDF$(v,k,\lambda)$. 
 
\smallskip 
 
We recall the notion of \emph{authentication perpendicular arrays}. These combinatorial structures are generalizations of Latin squares.

\begin{definition}\label{APA}
An \emph{authentication perpendicular array} \linebreak APA$_\lambda(t,k,v)$ is a $\lambda  {v \choose t} \times k$ array, $A$, of $v$ symbols, which satisfies the following properties:

\begin{enumerate}

\item[(i)] every row of $A$ contains $k$ distinct symbols,

\smallskip

\item[(ii)] for any $t$ columns of $A$, and for any $t$ distinct symbols, there are precisely $\lambda$ rows $r$ of $A$ such that the $t$ given symbols all occur in row $r$ in the given $t$ columns,

\smallskip

\item[(iii)] for any $s \leq t-1$ and for any $s+1$ distinct symbols $\{x_i\}_{i=1}^{s+1}$, it holds that among all the rows of $A$ that contain all the symbols $\{x_i\}_{i=1}^{s+1}$, the $s$ symbols $\{x_i\}_{i=1}^{s}$ occur in all possible subsets of $s$ columns equally often.

\end{enumerate}

\end{definition}

We present a simple example (due to van Rees, cf.~\cite{Stin93}):

\begin{example}
A $55 \times 3$ array $A$ can be constructed by developing the five rows
\[\begin{array}{ccc}
  0 & 1 & 2 \\
  0 & 9 & 7 \\
  0 & 3 & 6 \\
  0 & 4 & 8 \\
  0 & 5 & 10
\end{array}\]
modulo 11. Every pair $\{x_1,x_2\}$ occurs in three rows of $A$. Within these three rows, $x_1$ occurs once in each of the three columns, as does $x_2$. This gives an \linebreak APA$_1(2,3,11)$.
\end{example}

We recall furthermore the definition of \emph{combinatorial \mbox{$t$-designs}}.

\begin{definition}\label{Des}
A \mbox{$t$-$(v,k,\lambda)$} \emph{design} $\mathcal{D}$ is a pair \mbox{$(X,\mathcal{B})$}, which satisfies the following properties:

\begin{enumerate}

\item[(i)] $X$ is a set of $v$ elements, called \emph{points},

\smallskip

\item[(ii)] $\mathcal{B}$ is a family of \mbox{$k$-subsets} of $X$, called \emph{blocks},

\smallskip

\item[(iii)] every \mbox{$t$-subset} of $X$ is contained in exactly $\lambda$ blocks.

\end{enumerate}

\end{definition}

We will denote points by lower-case and blocks by upper-case Latin letters.
Via convention, let $b:=\left| \mathcal{B} \right|$ denote the number of blocks.
Throughout this work, `repeated blocks' are not allowed, that is, the same \mbox{$k$-subset}
of points may not occur twice as a block. If $t<k<v$ holds, then we speak of a \emph{non-trivial} \mbox{$t$-design}.
For historical reasons, a \mbox{$t$-$(v,k,\lambda)$ design} with
$\lambda =1$ is called a \emph{Steiner \mbox{$t$-design}} (sometimes also a \emph{Steiner system}).
If $\mathcal{D}=(X,\mathcal{B})$ is a \mbox{$t$-$(v,k,\lambda)$}
design with $t \geq 2$, and $x \in X$ arbitrary, then the
\emph{derived design} with respect to $x$ is
\mbox{$\mathcal{D}_x=(X_x,\mathcal{B}_x)$}, where $X_x = X
\backslash \{x\}$, \mbox{$\mathcal{B}_x=\{B \backslash \{x\} : x
\in B \in \mathcal{B}\}$}. In this case, $\mathcal{D}$ is also
called an \emph{extension} of $\mathcal{D}_x$.
Obviously, $\mathcal{D}_x$ is a $(t-1)$-\linebreak\mbox{$(v-1,k-1,\lambda)$}
design.

For the existence of \mbox{$t$-designs}, basic necessary
conditions can be obtained via elementary counting arguments (see,
for instance,~\cite{BJL1999}):

\begin{lemma}\label{s-design}
Let $\mathcal{D}=(X,\mathcal{B})$ be a \mbox{$t$-$(v,k,\lambda)$}
design, and for a positive integer $s \leq t$, let $S \subseteq X$
with $\left|S\right|=s$. Then the number of blocks containing
each element of $S$ is given by
\[\lambda_s = \lambda \frac{{v-s \choose t-s}}{{k-s \choose t-s}}.\]
In particular, for $t\geq 2$, a \mbox{$t$-$(v,k,\lambda)$} design is
also an \mbox{$s$-$(v,k,\lambda_s)$} design.
\end{lemma}

It is customary to set $r:= \lambda_1$ denoting the
number of blocks containing a given point. It follows

\begin{lemma}\label{Comb_t=5}
Let $\mathcal{D}=(X,\mathcal{B})$ be a \mbox{$t$-$(v,k,\lambda)$}
design. Then the following holds:
\begin{enumerate}

\item[{\em(a)}] $bk = vr.$

\smallskip

\item[\em{(b)}] $\displaystyle{{v \choose t}  \lambda = b  {k \choose t}.}$

\smallskip

\item[\em{(c)}] $r(k-1)=\lambda_2(v-1)$ for $t \geq 2$.

\end{enumerate}
\end{lemma}

The next result (cf.~\cite{Stin90}) uses \mbox{$t$-designs} in order to construct authentication perpendicular arrays. Further similar recursive constructions have been obtained in~\cite{Trun95}.

\begin{theorem}[Stinson--Teirlinck]\label{tdes_APA}
Suppose there is a \mbox{$t$-$(v,k,\lambda)$} design and an authentication perpendicular array APA$_{\lambda^\prime}(t,k,k)$, then there is an APA$_{\lambda \cdot \lambda^\prime}(t,k,v)$.
\end{theorem}

Concerning the existence of \mbox{$t$-designs}, a seminal result by Teirlinck~\cite{Teir1987} shows that there exist non-trivial \mbox{$t$-designs} for all possible values of $t$.

\begin{theorem}[Teirlinck]\label{Teirl}
For given integers $t$ and $v$ with $v \equiv t \; (\emph{mod} \; (t+1)!^{2t+1})$ and $v \geq t+1 >0$, there exists a
\mbox{$t$-$(v,t+1,(t+1)!^{2t+1})$} design.
\end{theorem}

Teirlinck's recursive construction methods are constructive. However, for a given $t$, they result in \mbox{$t$-designs}
with extremely large values for $v$ and $\lambda$. For example, the smallest parameters for the case $t=7$ are \mbox{$7$-$(40320^{15} + 7,8,40320^{15})$}. Until now no non-trivial \linebreak Steiner \mbox{$t$-design} with $t>5$ has been found. Highly regular examples have been proven not to exist (cf., e.g.,~\cite{Hu2008}). We refer the reader to~\cite{BJL1999,crc06} for encyclopedic accounts of key results in combinatorial design theory. Various connections of \mbox{$t$-designs} with coding and information theory can be found in a recent survey~\cite{Hu2009} (with many additional references therein).


\section{Constructions using Combinatorial Structures}\label{Known}

\subsection{Equiprobable Source Probability Distribution}

When the source states are known to be independent and equiprobable, authentication systems which are \linebreak $(t-1)$-fold secure against spoofing can be constructed via \mbox{$t$-designs} (cf.~\cite{Stin90,Sch86,DeS88}).

\begin{theorem}[De\,Soete--Sch\"{o}bi--Stinson]\label{general}
Suppose \linebreak there is a \mbox{$t$-$(v,k,\lambda)$} design. Then there is an authentication system for $k$ equiprobable source states, having $v$ messages and $\lambda  {v \choose t}/{k \choose t}$ encoding rules, that is $(t-1)$-fold secure against spoofing. Conversely, if there is an authentication system for $k$ equiprobable source states, having  $v$ messages and ${v \choose t}/{k \choose t}$ encoding rules, that is $(t-1)$-fold secure against spoofing, then there is a Steiner \mbox{$t$-$(v,k,1)$} design.
\end{theorem}

With a focus on optimal constructions, the above result has been modified in~\cite{Stin90} and generalized recently in~\cite{Hub2009} to include also the aspect of perfect secrecy. In particular, the first infinite classes of optimal perfect secrecy systems that achieve twofold security have been constructed in~\cite{Hub2009} as well as further optimal systems that offer $3$- and \mbox{$4$-fold} security against spoofing. We give in Table~\ref{t-des} all presently known optimal perfect secrecy systems that are \mbox{$t$-fold} secure against spoofing with $t \geq 1$ under equiprobable source probability distributions.

\begin{table}[b!]
\renewcommand{\arraystretch}{1.3}
\caption{Optimal perfect secrecy systems from Steiner designs that are \mbox{$t$-fold} secure against spoofing attacks}\label{t-des}

\begin{center}

\begin{tabular}{|c||c c c| c|c|}
  \hline
  $t$ & $k$ & $v$ & $b=b_{\mbox{\tiny{opt}}}$ &  \mbox{Ref.} \\
  \hline \hline
1   & $q+1$ & $\frac{q^{d+1}-1}{q-1}$ & $\frac{v(v-1)}{k(k-1)}$ &  \cite{Stin90}\\
            & $q$  {{\mbox{prime power}}} &  $d \geq 2$ {{\mbox{even}}} & & \\
   \hline
1   & $3$ & $v \equiv 1$ ({{\mbox{mod}}} $6$) & $\frac{v(v-1)}{6}$ & \cite{Hub2009}\\
   \hline
1  & $4$ & $v \equiv 1$ ({{\mbox{mod}}} $12$) & $\frac{v(v-1)}{12}$ & \cite{Hub2009}\\
   \hline
1   & $5$ & $v \equiv 1$ ({{\mbox{mod}}} $20$) & $\frac{v(v-1)}{20}$ & \cite{Hub2009}\\
   \hline
2   & $q+1$ & $q^d+1$ & $\frac{v(v-1)(v-2)}{k(k-1)(k-2)}$ & \cite{Hub2009}\\
            & $q$ {{\mbox{prime power}}} &  $d \geq 2$ {{\mbox{even}}} & & \\
   \hline
2  & $4$ & $v \equiv 2, 10$ ({{\mbox{mod}}} $24$) & $\frac{v(v-1)(v-2)}{24}$ &  \cite{Hub2009}\\
   \hline
2 & 5  & 26 & 260 &   \cite{Hub2009} \\
   \hline
   & 5  & 11 & 66 &   \cite{Hub2009} \\
   & 7  & 23 & 253  & \cite{Hub2009}\\
   & 5  & 23 & 1{,}771  & \cite{Hub2009}\\
   & 5  & 47 & 35{,}673  & \cite{Hub2009}\\
3 & 5  & 83 & 367{,}524  &  \cite{Hub2009}\\
   & 5  & 71 & 194{,}327 &  \cite{Hub2009}\\
   & 5  & 107 & 1{,}032{,}122 &  \cite{Hub2009}\\
   & 5  & 131 & 2{,}343{,}328 &  \cite{Hub2009}\\
   & 5  & 167 & 6{,}251{,}311 &   \cite{Hub2009}\\
   & 5  & 243 & 28{,}344{,}492 &  \cite{Hub2009} \\
   \hline
   & 6  & 12 & 132 &    \cite{Hub2009}\\
4  & 6  & 84 & 5{,}145{,}336  &  \cite{Hub2009}\\
   & 6  & 244 & 1{,}152{,}676{,}008  &  \cite{Hub2009} \\
  \hline
\end{tabular}
\end{center}
\end{table}

\subsection{Arbitrary Source Probability Distribution}\label{arbitrary}

For arbitrary source probability distributions, basically two construction methods have been developed for perfect secrecy systems that offer security against spoofing attacks (cf.~\cite{BE94,Bier06,Stin90,Trun95}). These constructions inherently require larger numbers of encoding rules for achieving the same level of security. One of the two methods with the smaller number of encoding rules requires $\lambda  {v \choose t}$ encoding rules when we want the perfect secrecy systems with $k$ source states and $v$ messages to be $(t-1)$-fold secure against spoofing (indeed, these systems achieve perfect $t$-fold secrecy), and is based on authentication perpendicular arrays APA$_{\lambda}(t,k,v)$, cf.~\cite[Thm.\,3.3]{Stin90}. For $t\geq 6$, there are --- apart from two infinite series with extremely large values of $\lambda$ --- only a very small number of authentication perpendicular arrays APA$_{\lambda}(t,k,v)$ known. These have been constructed via Theorem~\ref{tdes_APA} or similar results using \mbox{$t$-designs}. All these APA$_{\lambda}(t,k,v)$ have $t\leq 8$, and for $t=6$ all have $\lambda \geq 24$, for $t=7$ all have $\lambda \geq 70$, and for $t=8$ all have $\lambda \geq 280$. The two infinite series were constructed by Tran van Trung~\cite{Trun95} and have parameters $v \geq k$, $k=2t$ resp. $2t+1$, and $\lambda=t!^2 {v-t \choose t}/6!$ resp. $(t+1)t!^2 {v-t \choose t+1}/6!$.


\section{A General Construction Method \& Theoretical Point of View}\label{NewMethod}

We present a construction method for designing perfect secrecy systems that provide immunity to spoofing attacks under equiprobable source probability distributions. 

\begin{theorem}\label{mythm1}
Suppose there is a \mbox{$t$-$(v,k,\lambda)$} design, where $v$ divides the number of blocks $b=\lambda  {v \choose t}/{k \choose t}$. Then there is a perfect secrecy system for $k$ equiprobable source states, having $v$ messages and $b$ encoding rules, that is $(t-1)$-fold secure against spoofing. Moreover, the system is optimal if and only if $\lambda=1$.
\end{theorem}

\begin{proof}
Let $\mathcal{D}=(X,\mathcal{B})$ be a \mbox{$t$-$(v,k,\lambda)$} design, where  $v$ divides $b=\lambda  {v \choose t}/{k \choose t}$. It follows from Theorem~\ref{general} that the system is \mbox{$(t-1)$}-fold secure against spoofing attacks. Thus, it remains to verify that the system also achieves perfect secrecy when we assume that the encoding rules are used with equal probability. By Lemma~\ref{frequency}, this means that a given message must occur with the same frequency in each column of the resulting encoding matrix. This can be achieved by
ordering every block of $\mathcal{D}$ in such a way that every point occurs in each possible position in precisely $b/v$ blocks.
Since every point occurs in exactly $r=\lambda  {v-1 \choose t-1}/{k-1 \choose t-1}$ blocks in view of Lemma~\ref{Comb_t=5}~(c), necessarily $k$ must divide $r$. By Lemma~\ref{Comb_t=5}~(b), this is equivalent to saying that $v$ divides $b$.
To show that the condition is also sufficient, we may consider the bipartite point-block incidence graph of $\mathcal{D}$ with vertex set $X \cup \mathcal{B}$, where $(x,B)$ defines an edge if and only if $x \in B$  for $x \in X$ and $B \in \mathcal{B}$. An ordering on each block of $\mathcal{D}$ can be obtained via an edge-coloring of this graph using $k$ colors in such a way that each vertex $B \in \mathcal{B}$ is adjacent to one edge of each color, and each vertex $x \in X$ is adjacent to $b/v$ edges of each color. Technically, this can be achieved by first splitting up each vertex $x$ into $b/v$ copies, each having degree $k$, and then by finding an appropriate edge-coloring of the resulting $k$-regular bipartite graph using $k$ colors. We can now take the ordered blocks as encoding rules, each used with equal probability. Moreover, optimality occurs if and only if $\lambda=1$ in view of Theorem~\ref{thm_mas_sch}.\qed
\end{proof}

We note that the special case when $\lambda=1$ has been treated in~\cite[Thm.\,6]{Hub2009}.

Using Theorem~\ref{Teirl}, we may constructively generate systems that can reach an arbitrary high level of security against spoofing.

\begin{theorem}\label{mythm2}
For all integers $t$ and $v$ with $v \equiv t \; (\emph{mod} \;$ $(t+1)!^{2t+1})$ and $v \geq t+1 >0$,
there exists a perfect secrecy system for $t+1$ equiprobable source states, having $v$ messages and $b=(t+1)!^{2t}t! {v \choose t}$ encoding rules, that is $(t-1)$-fold secure against spoofing.
\end{theorem}

\begin{proof}
For the given design parameters, the division property $v \mid b$ holds:
\begin{align*}{}
v \mid \lambda \frac{{v \choose t}}{{k \choose t}}
& \Leftrightarrow  k(k-1) \cdots (k-t+1) \mid \lambda (v-1) \cdots (v-t+1)\\
& \Leftrightarrow  (t+1)! \mid (t+1)!^{2t+1}(v-1)\cdots (v-t+1).
\end{align*}
Therefore, the claim follows by applying Theorem~\ref{mythm1}.\qed
\end{proof}


\section{Explicit Constructions (I): Onefold Immunity}\label{one_immune}

We give in this section very efficient constructions of new optimal systems that are onefold secure against spoofing.

\begin{theorem}\label{mythm_CDF}
If there exists  a difference family \linebreak DF$(v,k,\lambda)$ over a finite additive Abelian group $G$ of order $v$, then there is a perfect secrecy system for $k$ equiprobable source states, having $v$ messages and  $b=\lambda  v(v-1)/(k^2 -k)$ encoding rules, that is onefold secure against spoofing. Moreover, the system is optimal if and only if $\lambda=1$.
\end{theorem}

\begin{proof}

Let $\mathcal{F}=\{D_1,\ldots,D_l\}$ be a DF$(v,k,\lambda)$ over $G$. We shall need the two basic facts:

\begin{itemize}

\item[$\bullet$] Since  $l=\frac{\lambda (v-1)}{k(k-1)}$ is a positive integer, we have
 \[\lambda (v-1) \equiv 0 \; (\mbox{mod} \; k(k-1))\quad (*).\] 

\item[$\bullet$] Let $ \mbox{Orb}_G(D_i)=\{D_i +g : g \in G\}$ denote the \emph{$G$-orbit} of $D_i$.
Then the union \[\bigcup_{i=1}^l \mbox{Orb}_G(D_i)\] forms the family of blocks of a \mbox{$2$-$(v,k,\lambda)$} design admitting $G$ as a group of automorphisms acting regularly (i.e., sharply transitively) on the points and semiregularly on the blocks.  
 
\end{itemize}
 
Thus, by $(*)$ and Lemma~\ref{Comb_t=5}, we have $v \mid b$, and the requirements for applying Theorem~\ref{mythm1} are fulfilled. \qed

\end{proof} 

In particular, when $\mathcal{F}=\{D_1,\ldots,D_l\}$ is a \linebreak CDF$(v,k,\lambda)$, then a perfect secrecy system can be constructed very efficiently due to the extremely simple form of its encoding matrix (cf.~Table~\ref{CDF_13}). We note that the special case when $l=1$ in the above theorem has been considered in~\cite[Thm.\,6.5\,\&\,Remark]{Stin90}.  In this case, the respective cyclic difference set is a \emph{Singer difference set} yielding a projective plane of prime power order as  \emph{symmetric} cyclic Steiner \mbox{$2$-design} (i.e., $v=b$). We give an example of a perfect secrecy systems constructed via Theorem~\ref{mythm_CDF} based on a CDF$(13,3,1)$.

\begin{example}\label{Ex2}
A CDF$(13,3,1)$ has two base blocks $D_1=\{0,1,4\}$ and $D_2=\{0,2,7\}$. The orbits of $D_1$ and $D_2$ immediately form an encoding matrix as given in Table~\ref{CDF_13}. The perfect secrecy system, having $3$ equiprobable source states,  $13$ messages and $26$ encoding rules, is optimal and offers onefold security against spoofing.
\end{example}

\begin{table}
\renewcommand{\arraystretch}{1.3}
\caption{Perfect secrecy system from a cyclic difference family CDF$(13,3,1)$.}\label{CDF_13}

\hspace*{0.1cm}

\begin{center}

\begin{tabular}{|c| c c c|}
  \hline
  & $s_1$ & $s_2$ &$ s_3$\\
  \hline
  $e_1$ &  0 & 1 & 4\\
  $e_2$ &  1 & 2 & 5 \\ 
  $e_3$ &  2 & 3 & 6 \\
  $e_4$ &  3 & 4 & 7\\  
  $e_5$ &  4 & 5 & 8\\   
  $e_6$ &  5& 6 & 9\\
  $e_7$ &  6& 7& 10\\ 
  $e_8$ &  7& 8& 11\\
  $e_9$ &  8&9 &12 \\  
  $e_{10}$ & 9 &10 &0 \\  
  $e_{11}$ &  10&11 & 1\\
  $e_{12}$ &  11& 12& 2\\ 
  $e_{13}$ &  12& 0& 3\\
  \hdashline
  $e_{14}$ &  0& 2 & 7\\  
  $e_{15}$ &  1& 3& 8\\
  $e_{16}$ &  2& 4& 9\\  
  $e_{17}$ &  3& 5& 10\\
  $e_{18}$ &  4& 6& 11\\ 
  $e_{19}$ &  5& 7& 12\\
  $e_{20}$ &  6& 8& 0\\  
  $e_{21}$ &  7& 9& 1\\
  $e_{22}$ &  8& 10& 2\\  
  $e_{23}$ &  9& 11& 3\\
  $e_{24}$ &  10& 12& 4\\ 
  $e_{25}$ &  11& 0& 5\\
  $e_{26}$ &  12& 1& 6\\     
  \hline
\end{tabular}

\end{center}
\end{table}

\begin{example}
The following infinite ((i)-(iii)) and finite ((iv)-(v)) families of cyclic difference families \linebreak $\mbox{CDF}(q,k,1)$ with $q$ a prime power are known (cf.~\cite{CWZ02} and the references therein;~\cite{crc06}):

\begin{enumerate}

\item[(i)] For $k=3,4$ and $5$, respectively, a $\mbox{CDF}(q,k,1)$ exists for all prime powers $q \equiv 1$ (mod $k(k-1)$).

\item[(ii)] A $\mbox{CDF}(q,6,1)$ exists for all prime powers $q \equiv 1$ (mod $30$) with the exception $q=61$.

\item[(iii)] A $\mbox{CDF}(q,7,1)$ exists for all prime powers $q \equiv 1$ (mod $42$) with the exception $q=43$, and the possible exceptions $q=127,\, 211,\,31^6$ as well a $q \in [261239791, \,1{.}236597 \times 10^{13}]$ such that $(-3)^{\frac{q-1}{14}} = 1$ in $\mathbb{F}_q$.

\item[(iv)] A $\mbox{CDF}(q,8,1)$ exists for all prime powers $q \equiv 1$ (mod $56$) $ < 10^4$, with the possible exceptions $q=113, \,169, \,281, \,337$.

\item[(v)] A $\mbox{CDF}(q,9,1)$ exists for all prime powers $q \equiv 1$ (mod $72$) $ < 10^4$, with the possible exceptions $q=289, \,361$.

\end{enumerate}

Hence, in all these cases a perfect secrecy system for $k$ equiprobable source states, having $q$ messages and $q(q-1)/(k^2-k)$ encoding rules, that is optimal and onefold secure against spoofing can be constructed very efficiently.

\end{example}

\section{Explicit Constructions (II): Multifold Immunity}\label{multi_immune}

We construct in this section the first near-optimal systems that are $5$- and $6$-fold secure as well as further systems with a feasible number of keys that are $7$-fold secure against spoofing. Recall that number of encoding rules in Theorem~\ref{mythm1} is $\lambda$ times the lower bound of Theorem~\ref{thm_mas_sch}. In order to construct perfect secrecy systems with a high level of security against spoofing, we are therefore interested in \mbox{$t$-designs} with large $t$ and small values of $\lambda$. These designs must satisfy the divisibility condition $v \mid b=\lambda  {v \choose t}/{k \choose t}$ of Theorem~\ref{mythm1}. When $2 \leq \lambda \leq 10$, we call such a system \emph{near-optimal}.

Relying on the Kramer--Mesner method~\cite{KM76}, various \mbox{$t$-designs} with large $t$ have been constructed in recent years under some prescribed groups of automorphisms (cf.~\cite{BLaue97,BLaue99,BLaue99b,Laue06}). We give some examples related to our considerations.

\begin{example}
A \mbox{$6$-$(19,7,4)$} design and three \mbox{$6$-$(19,7,6)$} designs have been constructed in~\cite{BLaue97} by prescribing the groups $Hol(C_{17})\!+\!+$ and  $Hol(C_{19})$, respectively (where the $+$ operator adds a fixed point to a permutation group).
The only known two smaller \mbox{$6$-$(14,7,4)$} designs have $C_{13}+$ as a prescribed group of automorphisms, but do not satisfy our divisibility condition. The only known further $6$-design with $\lambda=4$ has parameters \mbox{$6$-$(23,7,4)$}, and is derived from the unique \mbox{$7$-$(24,8,4)$} design with $PSL(2,23)$ as a prescribed group of automorphisms.
\end{example}

\begin{example}
There are \mbox{$7$-$(24,8,\lambda)$} designs admitting \linebreak $PSL(2,23)$ with possible values $\lambda=4,\ldots,8$. However, only for $\lambda=8$ the divisibility condition is fulfilled. There exist \mbox{$7$-$(26,8,6)$} designs, which have been constructed  with $PGL(2,25)$ as a prescribed group of automorphisms (cf.~\cite{BLaue97}).
\end{example}

\begin{example}
The construction of \mbox{$8$-$(31,10,100)$} designs has been established in~\cite{BLaue99b} with $PSL(3,5)$ as a prescribed group of automorphisms. The only known $8$-designs with smaller $\lambda$ are \mbox{$8$-$(31,10,93)$} designs admitting $PSL(3,5)$ again, but do not satisfy the divisibility condition.
\end{example}

We present in Table~\ref{6-des_NEW} all near-optimal perfect secrecy systems that are $5$- and \mbox{$6$-fold} secure against spoofing under equiprobable source probability distributions. We give the parameters of the systems as well as of the respective designs. We also indicate the optimal number $b_{\mbox{\footnotesize{opt}}}$ of encoding rules with respect to Theorem~\ref{thm_mas_sch}. All presently known \mbox{$t$-designs} with $t>5$ and $\lambda \leq 10$ have been considered. We generally remark that all known \mbox{$t$-$(v,k,\lambda)$} designs with $t >5$ have $\lambda \geq 4$. Furthermore, three infinite series of \mbox{$6$-designs} are known, however, for each $\lambda$ increases rapidly.

\begin{table}
\renewcommand{\arraystretch}{1.3}
\caption{Near-optimal perfect secrecy systems from $6$- and \mbox{$7$-designs} that are $5$- and \mbox{$6$-fold} secure against spoofing attacks}\label{6-des_NEW}

\begin{center}
\begin{tabular}{|c||c c c|c|c|}
  \hline
  $t$ & $k$ & $v$ & $b$ & $b_{\mbox{\tiny{opt}}}$ &\mbox{Design Parameters}\\
  \hline \hline
   & 7  & 19 & $4\times b_{\mbox{\tiny{opt}}}$ & 3{,}876 & $6$-$(19,7,4)$ \\
   & 7  & 22 & $8\times b_{\mbox{\tiny{opt}}}$ & 10{,}659 & $6$-$(22,7,8)$ \\
 5 & 7  & 23 & $4\times b_{\mbox{\tiny{opt}}}$ & 14{,}421 & $6$-$(23,7,4)$ \\
   & 7  & 25 & $6\times b_{\mbox{\tiny{opt}}}$ & 25{,}300 & $6$-$(25,7,6)$  \\
   & 7  & 32 & $6\times b_{\mbox{\tiny{opt}}}$ & 129{,}456 & $6$-$(32,7,6)$  \\
   \hline
   & 8  & 24 & $8\times b_{\mbox{\tiny{opt}}}$ & 43{,}263 & $7$-$(24,8,8)$  \\
6  & 8  & 26 & $6\times b_{\mbox{\tiny{opt}}}$ & 82{,}225 & $7$-$(26,8,6)$  \\
   & 8  & 33 &$10\times b_{\mbox{\tiny{opt}}}$ & 534{,}006 & $7$-$(33,8,10)$  \\
   \hline
\end{tabular}
\end{center}
\end{table}

\begin{table}
\renewcommand{\arraystretch}{1.3}
\caption{Some perfect secrecy systems from \mbox{$8$-designs} that are \mbox{$7$-fold} secure against spoofing attacks}\label{8-des_NEW}

\begin{center}
\begin{tabular}{|c||c c c|c|c|}
  \hline
  $t$ & $k$ & $v$ & $b$ & $b_{\mbox{\tiny{opt}}}$ &\mbox{Design Parameters}\\
  \hline \hline
   & 10  & 31 & $100\times b_{\mbox{\tiny{opt}}}$ & 175{,}305 & $8$-$(31,10,100)$  \\
   & 11  & 27 & $432\times b_{\mbox{\tiny{opt}}}$ & 13{,}455 & $8$-$(27,11,432)$  \\
 7 & 11  & 36 &$1{,}260\times b_{\mbox{\tiny{opt}}}$ & 183{,}396 & $8$-$(36,11,1260)$  \\
   & 11  & 40 &$1{,}440\times b_{\mbox{\tiny{opt}}}$ & 466{,}089 & $8$-$(40,11,1440)$  \\
   & 12  & 27 &$1{,}296\times b_{\mbox{\tiny{opt}}}$ & 4{,}485 & $8$-$(27,12,1296)$  \\
   \hline
\end{tabular}
\end{center}
\end{table}

\smallskip

In Table~\ref{8-des_NEW}, we give further perfect secrecy systems with a feasible number of encoding rules that are \mbox{$7$-fold} secure against spoofing under equiprobable source probability distributions. All presently known \mbox{$t$-designs} with $t>7$ and $\lambda \leq 3{,}000$ have been considered.

\smallskip

We refer to the above references for further information on the respective designs.

\begin{remark}
As indicated in Table~\ref{6-des_NEW}, a perfect secrecy system, constructed from a $6$-$(23,7,4)$ design, with \mbox{$k=7$} equiprobable source states and $v=23$ messages that is \mbox{$5$-fold} secure against spoofing requires $57{,}684$ encoding rules. A perfect secrecy system, constructed from a \mbox{$6$-$(25,7,6)$ design},  with $k=7$ equiprobable source states and $v=25$ messages that is \mbox{$5$-fold} secure against spoofing requires $151{,}800$ encoding rules.

For comparison, a perfect (\mbox{$5$-fold}) secrecy system, constructed from an \mbox{APA$_{10}(5,6,24)$}, with $k=6$ source states and $v=24$ messages that offers \mbox{$4$-fold} security against spoofing for an arbitrary source probability distribution requires $425{,}040$ encoding rules. A perfect (\mbox{$5$-fold}) secrecy system, constructed from an \linebreak \mbox{APA$_{60}(5,7,24)$}, with $k=7$ source states and $v=24$ messages that is \mbox{$4$-fold} secure against spoofing for an arbitrary source probability distribution requires $2{,}550{,}240$ encoding rules (cf.~Subsection~\ref{arbitrary}).
\end{remark}


\section{Application to the Verification Oracle Model}\label{oracle}

We will now consider the scenario, where the opponent has access to a \emph{verification oracle (V-oracle)}.
In this extended authentication model, we assume that the opponent is no longer restricted to \emph{passively} observing messages transmitted by the sender to the receiver. The opponent may send a message of the opponent's choice to the receiver and observe the receiver's response whether or not the receiver accepts it as authentic. This more powerful, \emph{pro-active} attack scenario can be modeled in terms of a V-oracle that provides a response (\emph{accept} or \emph{reject}) to a query message in the same way as the message would be accepted or not by the legitimate receiver. This attack model was recently introduced in~\cite{gold04,saf04}. We recall and slightly adjust the notation as far as it is necessary for our consideration. Further details on this model can be found in~\cite{gold04,saf04,tonien07,tonien09}.

In~\cite{tonien07}, the two types of \emph{online} and \emph{offline} attacks are studied. In the online attack, the receiver is supposed to respond to each incoming query message, and thus the opponent is successful as soon as the receiver accepts a message as authentic. Thus, every query message is at the same time a spoofing message. In the offline attack, the query and the spoofing phase are separated. First, the opponent makes all his queries to the oracle, and then uses this collected (state) information to construct a spoofing message. In both scenarios, the opponent is assumed to be adaptive. The online attack  models an opponent's interaction with a verification oracle such as an ATM machine, while in the offline attack the opponent may have captured an offline verification box. Often, the offline attack model is used as an intermediate model for analyzing the online scenario.
We speak in each scenario of a \emph{spoofing attack} of \mbox{order $i$} \emph{in the V-oracle model} if the opponent has access to $i$ verification queries. The opponent's strategy can be modeled via probability distributions on the query set $\mathcal{M}$ of verification queries. The \emph{online deception probability} $P_{d_i}^{\mbox{\tiny{\sf{online}}}}$, respectively \emph{offline deception probability} $P_{d_i}^{\mbox{\tiny{\sf{offline}}}}$, denotes the probability that the opponent can deceive the receiver with a spoofing attack of order $i$. In~\cite{tonien07}, lower bounds on these deception probabilities have been obtained.

\begin{theorem}[Tonien--Safavi-Naini--Wild]
In an authentication system with $k$ source states and $v$ messages, the offline and online deception probabilities in the V-oracle model are bounded below by
\[P_{d_i}^{\mbox{\tiny{\sf{offline}}}} \geq \frac{k}{v} \quad \mbox{and} \quad P_{d_i}^{\mbox{\tiny{\sf{online}}}} \geq 1- \frac{{v-k \choose i+1}}{{v \choose i+1}},\quad \mbox{respectively}.\]
\end{theorem}

Interestingly, it furthermore follows that
\[P_{d_i}^{\mbox{\tiny{\sf{offline}}}} = \frac{k}{v} \quad \mbox{if and only if} \quad P_{d_i}^{\mbox{\tiny{\sf{online}}}} = 1-\frac{{v-k \choose i+1}}{{v \choose i+1}}.\]

Thus, an authentication system that attains the \linebreak bound in the offline attack is the same as in the online attack, and vice versa.
Clearly, $P_{d_i}^{\mbox{\tiny{\sf{offline}}}}$ is independent of $i$. If the bound for $P_{d_i}^{\mbox{\tiny{\sf{online}}}}$ is satisfied with equality, then also the bound for $P_{d_{i-1}}^{\mbox{\tiny{\sf{online}}}}$ is satisfied with equality for $i>1$ (cf.~\cite{tonien07}). Hence, we call a system \emph{\mbox{$t$-fold} secure against spoofing in the V-oracle model} if $P_{d_t}^{\mbox{\tiny{\sf{offline}}}} = \frac{k}{v}$ or, equivalently, $P_{d_t}^{\mbox{\tiny{\sf{online}}}} = 1-\frac{{v-k \choose t+1}}{{v \choose t+1}}$. The notation of perfect secrecy holds as given in Section~\ref{Model}. An analogue to Theorem~\ref{thm_mas_sch} has been derived in~\cite{tonien07} for the V-oracle model.

\begin{theorem}[Tonien--Safavi-Naini--Wild]\label{masssch_oracle}
If an \linebreak authentication system is $(t-1)$-fold secure against spoofing in the V-oracle model, then the number of encoding rules is bounded below by
\[b \geq \frac{{v \choose t}}{{k \choose t}}.\]
\end{theorem}

Again, we call a system \emph{optimal} when the lower bound holds with equality. For equiprobable source \linebreak states, optimal authentication systems which are $(t-1)$-fold against spoofing in the V-oracle model have been characterized in~\cite{tonien07}. We give the result in a slightly more generalized form, which can easily be obtained from the original proof.

\begin{theorem}[Tonien--Safavi-Naini--Wild]\label{DeSoete_oracle}
Suppose there is a \mbox{$t$-$(v,k,\lambda)$} design. Then there is an authentication system for $k$ equiprobable source states, having $v$ messages and $\lambda \cdot {v \choose t}/{k \choose t}$ encoding rules, that is \mbox{$(t-1)$-fold} secure against spoofing in the V-oracle model. Conversely, if there is an authentication system for $k$ equiprobable source states, having  $v$ messages and ${v \choose t}/{k \choose t}$ encoding rules, that is \mbox{$(t-1)$-fold} secure against spoofing in the V-oracle model, then there is a Steiner \linebreak \mbox{$t$-$(v,k,1)$} design.
\end{theorem}

We will apply now Theorem~\ref{mythm1} to construct perfect secrecy systems that provide a high level of security against spoofing in the V-oracle model for equiprobable source probability distributions. This generalizes the result~\cite[Thm.\,3.27]{Hu2010now}, where the case $\lambda=1$ has been treated.

\begin{theorem}\label{mythm1_oracle}
Suppose there is a \mbox{$t$-$(v,k,\lambda)$} design, where $v$ divides the number of blocks $b=\lambda  {v \choose t}/{k \choose t}$. Then there is a perfect secrecy system for $k$ equiprobable source states, having $v$ messages and $b$ encoding rules, that is $(t-1)$-fold secure against spoofing in the V-oracle model. Moreover, the system is optimal if and only if $\lambda=1$.
\end{theorem}

\begin{proof}
By Theorem~\ref{DeSoete_oracle}, the system is $(t-1)$-fold secure against spoofing in the V-oracle model. Under the assumption that the encoding rules are used with equal probability, we may proceed as in the proof of Theorem~\ref{mythm1} to verify that the system also achieves perfect secrecy. With respect to Theorem~\ref{masssch_oracle} optimality is obtained if and only if $\lambda=1$.\qed
\end{proof}

Clearly, Theorem~\ref{mythm2} can also be applied to the V-oracle model.

\begin{theorem}\label{mythm2_oracle}
For all integers $t$ and $v$ with  $v \equiv t \; (\emph{mod} \;$ $(t+1)!^{2t+1})$ and $v \geq t+1 >0$,
there exists a perfect secrecy system for $t+1$ equiprobable source states, having $v$ messages and $b=(t+1)!^{2t}t! {v \choose t}$ encoding rules, that is $(t-1)$-fold secure against spoofing in the V-oracle model.
\end{theorem}

All the results in Section~\ref{one_immune} and Section~\ref{multi_immune} may be transferred accordingly.


\section{Conclusion}\label{Conl}

We have given novel perfect secrecy systems that provide immunity to spoofing attacks under equiprobable source probability distributions. Our construction method generalized in a natural manner the approach in~\cite{Hub2009} and allowed us to use \mbox{$t$-designs} instead of merely Steiner \mbox{$t$-designs} in the construction process. From a theoretical point of view, we have shown that based on Teirlinck's existence result for $t$-designs,  perfect secrecy systems can be generated that can reach an arbitrary high level of security. Concerning explicit constructions, we have obtained, via cyclic difference families, very efficient constructions of new optimal systems that are onefold secure against spoofing. By using \mbox{$t$-designs} for large values of $t$, we have also presented the first near-optimal systems that are $5$- and $6$-fold secure as well as further systems with a feasible number of keys that are $7$-fold secure against spoofing.  Previous constructions of multifold secure systems had been known only for arbitrary source probability distributions, which inherently result in larger numbers of encoding rules for achieving the same level of security. We have furthermore applied our results to a recently extended authentication model, where the opponent has access to a verification oracle. Novel perfect secrecy systems with immunity to spoofing in the verification oracle model have been obtained this way.

\end{document}